\documentclass{ifacconf}

\usepackage{algorithm}
\usepackage{algpseudocode}
\usepackage{amsmath} 
\usepackage{amssymb} 
\usepackage{natbib}        
\setcitestyle{square}

\newtheorem{theorem}{Theorem}      
\newtheorem{lemma}{Lemma}

\newtheorem{definition}{Definition} 
\newtheorem{assumption}{Assumption}

\usepackage{graphicx}          
\begin{document}
\begin{frontmatter}

\title{Split-Merge Dynamics for Shapley-Fair Coalition Formation} 

\author[NYU]{Quanyan Zhu}
\author[NYU]{Zhengye Han} 

\address[NYU]{Department of Electrical and Computer Engineering, \\
              New York University, Brooklyn, NY 11201 USA \\
              (e-mail: qz494@nyu.edu, zh3286@nyu.edu).}

\begin{abstract}
Coalition formation is often modeled as a static equilibrium problem, neglecting the dynamic processes governing how agents self-organize. This paper proposes a dynamic split-and-merge framework that balances two conflicting economic forces: individual fairness and collective efficiency. We introduce a control-theoretic mechanism where topological operations are driven by distinct signals: splits are triggered by fairness violations (specifically, negative Shapley values representing ``agent-responsible inefficiency''), while merges are driven by strict surplus improvements (superadditivity). We prove that these dynamics converge in finite time to a specific class of steady states termed Shapley-Fair and Merge-Stable (SFMS) partitions. Convergence is established via a vector Lyapunov function tracking aggregate fairness deficits and system surplus, leveraging a discrete-time LaSalle invariance principle. Numerical case studies on a 10-player game demonstrate the algorithm's ability to resolve fairness tensions and reach stable configurations, providing a rigorous foundation for endogenous coalition formation in dynamic environments.
\end{abstract}

\begin{keyword}
Large Scale Systems, Stability, Optimization: Theory and Algorithms.
\end{keyword}

\end{frontmatter}
\section{Introduction}
Coalitions are the fundamental building blocks of socio-technical systems, orchestrating cooperation in everything from international treaties to distributed agent networks. However, these structures are rarely static; they form, dissolve, and reorganize in response to evolving incentives. While cooperative game theory has provided a rich vocabulary for describing stable outcomes—such as the core concept introduced in foundational works \citep{gillies1953some} or the partition-based stability in hedonic games \citep{bogomolnaia2002stability}—these classical frameworks are fundamentally static. They characterize stable configurations but remain silent on the dynamic pathways by which agents, starting from disarray, self-organize to reach such equilibria.

This limitation has motivated a shift towards dynamic perspectives. Seminal studies \citep{filar2000dynamic} and later robustness analyses \citep{bauso2009robust, smyrnakis2019game} introduced time-varying characteristic functions, while other approaches dealing with stochasticity \citep{timmer2004three} or myopic deviations \citep{konishi2003coalition} framed coalition formation as an evolving process. Yet, establishing a unified control-theoretic framework that explains how stable structures emerge from local rules remains a challenge. Existing generic approaches to split-and-merge dynamics, such as those discussed in \citep{apt2009generic}, provide a structural foundation but often lack a specific economic mechanism to resolve the inherent tension between efficiency (collective gain) and fairness (individual reward). As noted in bargaining dynamics research \citep{hart1996bargaining}, value allocation is intrinsic to stability; without perceived fairness, even efficient coalitions are prone to fragmentation.

In this work, we propose a dynamic framework that harmonizes these forces by assigning distinct economic roles to topological operations. Unlike approaches that treat splits and merges symmetrically, our algorithm utilizes the Shapley value as a strict filter for stability. We posit that a negative Shapley value signals ``agent-responsible inefficiency,'' triggering a fairness-driven split specifically designed to isolate detrimental agents. This targeted approach addresses the structural complexities of stability highlighted in \citep{faigle2001computation} by localizing corrective actions. Conversely, efficiency-driven merges are permitted only when they yield a strict surplus improvement (superadditivity) and, crucially, only after internal fairness has been restored. This ensures that efficiency gains do not come at the expense of individual stability.

We show that the resulting dynamics converge to a well-defined class of coalition structures, which we term \textbf{Shapley-Fair and Merge-Stable (SFMS)} partitions. An SFMS partition satisfies two properties: (i) Shapley fairness, meaning that no agent within any coalition receives a negative Shapley value in the restricted game induced by that coalition; and (ii) merge stability, meaning that no pair of coalitions can merge to obtain a strictly higher total surplus. These partitions characterize the steady states of the algorithm under the strongest fairness-enforcing split rule and refine classical notions of coalition stability by explicitly incorporating individual fairness alongside collective efficiency. To establish this convergence, we construct a vector Lyapunov function tracking aggregate fairness violations and system surplus. Leveraging a discrete-time LaSalle invariance principle, we prove that every trajectory enters, in finite time, an invariant set composed of SFMS fixed points or value-preserving limit cycles.

The paper is structured as follows. Section \ref{sec:preliminaries} introduces the game-theoretic preliminaries and the Shapley value. Section \ref{sec:dynamics} details the proposed merge-split algorithm. Section \ref{sec:invariance_principle} establishes the general invariance principle for coalition structures. Section \ref{sec:algorithm_invariance} applies this theory to analyze the convergence of the proposed algorithm. Section \ref{sec:case_study} presents a numerical case study demonstrating the dynamics, and Section \ref{sec:conclusion} concludes the paper.

\section{Preliminaries}
\label{sec:preliminaries}

Let $N=\{1,2, \ldots, n\}$ denote a finite set of players. A transferable-utility (TU) cooperative game is defined by a characteristic function $v: 2^{N} \rightarrow \mathbb{R},$ where $v(S)$ represents the total surplus generated by coalition $S \subseteq N$, with the convention $v(\varnothing)=0$. A partition (or coalition structure) of $N$ is a collection $\Pi=\left\{S_{1}, S_{2}, \ldots, S_{k}\right\}$ of disjoint, nonempty coalitions whose union is $N$. The set of all partitions of $N$ is denoted by $\mathcal{P}(N)$. A partition specifies how players are organized into coalitions that jointly generate and distribute surplus.

For a coalition $S \subseteq N$, we denote by $v|_S$ the restriction of the characteristic function to subsets of $S$ (i.e., $v|_S(T) := v(T), T \subseteq S$).

This restriction represents the induced cooperative game faced by the members of $S$ when acting independently of the rest of the population.

\subsection{Shapley value}
The Shapley value assigns to each player $i \in S$ a payoff $\phi_{i}\left(S,\left.v\right|_{S}\right)$ representing the player's expected marginal contribution to coalition $S$. It is defined as

\begin{equation*}
\phi_{i}\left(S,\left.v\right|_{S}\right)=\sum_{T \subseteq S \backslash\{i\}} \frac{|T|!(|S|-|T|-1)!}{|S|!}[v(T \cup\{i\})-v(T)] .
\end{equation*}

In particular, a player with $\phi_{i}\left(S,\left.v\right|_{S}\right)<0$ is one whose expected marginal contribution is negative: on average, the presence of that player reduces the surplus available to the coalition. Such players are naturally interpreted as being treated unfairly by the coalition structure or as imposing a net burden on the coalition. While the exact computation of Shapley values is combinatorially expensive ($NP$-hard), efficient approximation methods such as sampling permutations \citep{mitchell2022sampling} or Bayesian Monte Carlo techniques \citep{touati2021bayesian} can be employed. This ensures that the proposed fairness metrics remain computationally feasible even for larger player sets.

A fundamental property of the Shapley value is efficiency $\sum_{i \in S} \phi_{i}(S, v|_S)=v(S)$,  which states that the total surplus of a coalition is fully allocated among its members. This implies that negative Shapley values cannot occur in isolation: whenever some players receive negative allocations, others must receive positive allocations that compensate for these deficits. Thus, fairness violations are inherently a redistributive phenomenon within coalitions. Furthermore, due to the linearity of Shapley value, the ensuing analysis is invariant under affine transformations of the characteristic function $v$, allowing us to normalize payoffs without loss of generality.

\subsection{Negative Shapley mass and notation}
To quantify fairness violations, we introduce the positive-part operator. For any real number $x \in \mathbb{R}$, define $(x)^{+}:=\max \{x, 0\}$. For a coalition $R \subseteq N$, the negative Shapley mass of player $j \in R$ is given by $\left(-\phi_{j}\left(R,\left.v\right|_{R}\right)\right)^{+}$, which measures the extent to which player $j$ is under-allocated relative to zero. Summing this quantity over all players in a coalition yields the total fairness deficit within that coalition $\sum_{j \in R}\left(-\phi_{j}\left(R,\left.v\right|_{R}\right)\right)^{+}$.

This notion of negative Shapley mass serves as the fundamental fairness signal driving the split phase of the merge-split algorithm and forms the basis of the Lyapunov functions used to analyze convergence and invariant behavior.

\section{Merge-Split Dynamics}
\label{sec:dynamics}
We consider a cooperative game defined over a finite set of players $N=\{1,2, \ldots, n\}$ with characteristic function $v: 2^{N} \rightarrow \mathbb{R}$. Our objective is to construct coalition structures that balance individual Shapley fairness with collective surplus efficiency. The algorithm generates a sequence of partitions $\{\Pi_{t}\}_{t \geq 0}$ of the player set $N$ starting from an arbitrary initial partition $\Pi_0$. Each iteration consists of a parallel \textit{split phase} ($F^{\mathrm{S}}$), which enforces internal fairness, followed by a sequential \textit{merge phase} ($F^{\mathrm{M}}$), which exploits cross-coalition surplus.

\subsection{Split Operator $F^{\mathrm{S}}$}
The split phase refines coalitions exhibiting fairness violations. For any coalition $S \subseteq N$, we classify agents based on the sign of their Shapley values:
\[
P_{+}(S):=\left\{i \in S \mid \phi_{i}(S, v|_{S}) \geq 0\right\}, \quad P_{-}(S):=S \setminus P_{+}(S). 
\]
We employ a \textit{Fairness-Stable} split rule that explicitly enforces the internal consistency of the residual group.

\textbf{Split Rule (Fairness-stable residual coalition).} Define the refinement $\mathcal{S}(S)$ as:
\[
\mathcal{S}(S):= \begin{cases}
    \{\{i\} \mid i \in S\}, & \hspace{-0.8em} 
    \begin{aligned}[t] 
        \text{if } \sum_{j \in R}(-\phi_{j} (R,  v|_{R}))^{+} \\ 
        > 0 
    \end{aligned} \\
    \{R\} \cup\left\{\{i\} \mid i \in P_{-}(S)\right\}, & \text { otherwise }
\end{cases} 
\]
where $R:=P_{+}(S)$. Under this rule, the non-negative candidate coalition $R$ is retained only if it is internally Shapley-fair; otherwise, the entire coalition $S$ disintegrates into singletons. This strict condition prevents the regeneration of fairness violations, ensuring monotonicity of the Lyapunov function.

The global split operator $F^{\mathrm{S}}$ applies this rule parallelly to all coalitions: $F^{S}(\Pi):=\bigcup_{S \in \Pi} \mathcal{S}(S).$

\subsection{Merge Operator $F^{\mathrm{M}}$}
The merge phase sequentially aggregates coalitions to capture surplus. For any disjoint $A, B \in \Pi$, define the merge surplus $\Delta(A, B):=v(A \cup B)-(v(A)+v(B))$.
The operator $F^{\mathrm{M}}$ iteratively merges any pair with $\Delta(A, B) > 0$ until no such profitable pairs remain. The order of merges is arbitrary; thus $F^{\mathrm{M}}$ is a set-valued map if multiple profitable pairs exist.

\subsection{Steady States: SFMS Partitions}
The full dynamics iterate the split and merge phases: $\Pi_{t+1} \in F^{\mathrm{M}}(F^{\mathrm{S}}(\Pi_{t}))$. The fixed points of this system correspond to a specific class of stable structures.

\begin{definition}[SFMS Partition]
A partition $\Pi$ is \textit{Shapley-Fair and Merge-Stable} (SFMS) if:
\begin{itemize}
    \item Shapley Fairness: $\phi_{i}(S, v|_{S}) \geq 0$ for all $S \in \Pi$ and $i \in S$.
    \item Merge Stability: $v(A \cup B) \leq v(A)+v(B)$ for all distinct $A, B \in \Pi$.
\end{itemize}
\end{definition}

Under the Fairness-Stable split rule, SFMS partitions characterize precisely the steady states: internal fairness prevents splits, and superadditivity violation prevents merges. 

\begin{algorithm}
\label{algo:Split Rule}
\caption{Split Rule: Fairness-Stable Residual Coalition}
\begin{algorithmic}[1] 
\Require Coalition $S \subseteq N$, characteristic function $v$
\Ensure A collection of subcoalitions $\mathcal{S}(S)$

\State Compute Shapley values $\phi_i(S, v|_S)$ for all $i \in S$
\State Form sign-based subsets:
\[
    P_+(S) \leftarrow \{i \in S \mid \phi_i(S, v|_S) \geq 0\}, \quad P_-(S) \leftarrow S \setminus P_+(S)
\]

\If{$P_-(S) = \emptyset$}
    \State \textbf{return} $\{S\}$
\EndIf

\State Let $R \leftarrow P_+(S)$
\State Compute Shapley values $\phi_j(R, v|_R)$ for all $j \in R$
\State Evaluate the residual fairness violation:
\[
    \Theta(R) \leftarrow \sum_{j \in R} (-\phi_j(R, v|_R))^+
\]

\If{$\Theta(R) > 0$}
    \State \textbf{return} $\{\{i\} \mid i \in S\}$
\Else
    \State \textbf{return} $\{R\} \cup \{\{i\} \mid i \in P_-(S)\}$
\EndIf

\end{algorithmic}
\end{algorithm}

Unlike classical coalition formation rules that split whenever a sub-partition yields higher total welfare ($v(S) < \sum v(T_k)$), our approach is explicitly agent-responsible. Splits are triggered solely by negative Shapley values, targeting specific agents causing inefficiency. This avoids unnecessary structural reorganization when the inefficiency is purely due to coordination costs rather than individual burdens.

\section{Invariance Principle for Coalition Structures}
\label{sec:invariance_principle}
\begin{definition}[Lexicographic order]
 For $u, v \in \mathbb{R}^{k}$, we say that $u$ is lexicographically greater than or equal to $v$, written $u \succeq v$, if either $u=v$, or there exists an index $j \in\{1, \ldots, k\}$ such that $u_{i}=v_{i}$ for all $i<j$ and $u_{j}>v_{j}$. In other words, $u$ and $v$ agree in all earlier components, and $u$ is larger at the first component where they differ. We write $u \succ v$ if $u \succeq v$ and $u \neq v$.    
\end{definition}

\begin{lemma}[Finite lexicographic order has no infinite ascent]
\label{lemma:lexicographic order}
Let $Y \subset \mathbb{R}^{k}$ be a finite set, and consider the strict lexicographic order $\succ$ restricted to $Y$. Then there is no infinite strictly increasing sequence in $Y$. Equivalently, any sequence $\left\{y_{t}\right\}_{t \geq 0} \subset Y$ that is nondecreasing with respect to $\succeq$, i.e., $y_{t+1} \succeq y_{t}$ for all $t$, must become constant after finitely many steps: there exists $T<\infty$ such that $y_{t}=y_{T}$ for all $t \geq T$.    
\end{lemma}

\begin{pf}
Because $Y$ is finite, a strictly increasing sequence in ( $Y, \succ$ ) cannot repeat elements and therefore can have length at most $|Y|$. Consequently, a nondecreasing sequence can experience only finitely many strict increases. Once no further strict increase is possible, the sequence must remain constant thereafter.
\end{pf}

Let $N$ be a finite player set and $X:=\mathcal{P}(N)$ the finite set of all partitions of $N$. Let $F: X \rightrightarrows X$ be a set-valued map. A (forward) trajectory is any sequence $\left\{\Pi_{t}\right\}_{t \geq 0}$ such that $\Pi_{t+1} \in F\left(\Pi_{t}\right)$ for all $t$.

\begin{definition}[Weak invariance]
A set $S \subseteq X$ is weakly invariant under $F$ if $\forall \Pi \in S, \exists \Pi^{\prime} \in F(\Pi) \cap S$. The largest weakly invariant subset of a set $E \subseteq X$ is the union of all weakly invariant subsets of $E$ (it exists since $X$ is finite).    
\end{definition}

\begin{theorem}[Generalized discrete LaSalle invariance principle]
 Let $X=\mathcal{P}(N)$ be finite and $F: X \rightrightarrows X$. Let $\mathcal{V}: X \rightarrow \mathbb{R}^{k}$ and equip $\mathbb{R}^{k}$ with the lexicographic order. Assume:\\
(i) Monotonicity along all admissible moves: for all $\Pi \in X$ and all $\Pi^{\prime} \in F(\Pi), \mathcal{V}\left(\Pi^{\prime}\right) \succeq \mathcal{V}(\Pi)$.\\
(ii) Dichotomy on each state: for every $\Pi \in X$, either (a) there exists $\Pi^{\prime} \in F(\Pi)$ with $\mathcal{V}\left(\Pi^{\prime}\right)= \mathcal{V}(\Pi)$, or else (b) for all $\Pi^{\prime} \in F(\Pi), \mathcal{V}\left(\Pi^{\prime}\right) \succ \mathcal{V}(\Pi)$.

Define the zero-progress set $\mathcal{E} := \{\Pi \in X \mid \exists \Pi' \in F(\Pi) \text{ s.t. } \mathcal{V}(\Pi') = \mathcal{V}(\Pi)\}$. Let $\mathcal{I}$ be the largest weakly invariant subset of $\mathcal{E}$.
Then for every trajectory $\left\{\Pi_{t}\right\}_{t \geq 0}$ with $\Pi_{t+1} \in F\left(\Pi_{t}\right)$, there exists $T<\infty$ such that $\Pi_{t} \in \mathcal{I}$ for all $t \geq T$. In particular, the tail of the trajectory is eventually periodic, i.e., it is either a fixed point or a finite limit cycle contained in $\mathcal{I}$.    
\end{theorem}

\begin{pf}
Fix any trajectory $\left\{\Pi_{t}\right\}_{t \geq 0}$ with $\Pi_{t+1} \in F\left(\Pi_{t}\right)$.

Step 1: finite-time stabilization of $\mathcal{V}\left(\Pi_{t}\right)$. Because $X$ is finite, $Y:=\mathcal{V}(X)$ is finite. By (i), $\mathcal{V}\left(\Pi_{t+1}\right) \succeq \mathcal{V}\left(\Pi_{t}\right)$ for all $t$, hence $\left\{\mathcal{V}\left(\Pi_{t}\right)\right\} \subset Y$ is lexicographically nondecreasing. By Lemma \ref{lemma:lexicographic order}, there exists $T_{0}<\infty$ such that $\mathcal{V}\left(\Pi_{t}\right)=\mathcal{V}\left(\Pi_{T_{0}}\right)$ for all $t \geq T_{0}$.

Step 2: eventual membership in $\mathcal{E}$. Fix any $t \geq T_{0}$. Since $\Pi_{t+1} \in F\left(\Pi_{t}\right)$ and $\mathcal{V}\left(\Pi_{t+1}\right)= \mathcal{V}\left(\Pi_{t}\right)$, it follows that $\Pi_{t} \in \mathcal{E}$ by definition. Hence $\Pi_{t} \in \mathcal{E}$ for all $t \geq T_{0}$.

Step 3: restriction to a constant- $\mathcal{V}$ level set. Let $v^{\star}:=\mathcal{V}\left(\Pi_{T_{0}}\right)$ and define the level set $L:= \left\{\Pi \in \mathcal{E}: \mathcal{V}(\Pi)=v^{\star}\right\}$. Then $\Pi_{t} \in L$ for all $t \geq T_{0}$. Moreover, for any $\Pi \in L$ and any $\Pi^{\prime} \in F(\Pi)$, assumption (i) gives $\mathcal{V}\left(\Pi^{\prime}\right) \succeq v^{\star}$, while $\Pi \in \mathcal{E}$ ensures there exists some successor with equality. Assumption (ii) therefore rules out $\mathcal{V}\left(\Pi^{\prime}\right) \succ v^{\star}$ for a successor chosen along the stabilized trajectory; in particular, the trajectory uses only transitions that keep $\mathcal{V}$ equal to $v^{\star}$.

Define the restricted transition relation on $L: \Pi \rightarrow \Pi^{\prime}$ if $\Pi^{\prime} \in F(\Pi)$ and $\mathcal{V}\left(\Pi^{\prime}\right)=\mathcal{V}(\Pi)$. The tail $\left\{\Pi_{t}\right\}_{t \geq T_{0}}$ is a walk on the finite directed graph $G_{L}=(L, \rightarrow)$.

Step 4: eventual periodicity and entry into $\mathcal{I}$. Since $L$ is finite and $\left\{\Pi_{t}\right\}_{t \geq T_{0}}$ is an infinite walk in $G_{L}$, there exist integers $T_{1} \geq T_{0}$ and $p \geq 1$ such that $\Pi_{T_{1}}=\Pi_{T_{1}+p}$ (pigeonhole principle). Let $C:=\left\{\Pi_{T_{1}}, \Pi_{T_{1}+1}, \ldots, \Pi_{T_{1}+p-1}\right\}$. By construction, for every $\Pi \in C$ there exists a successor in $C$ (namely the next state on the cycle), so $C$ is weakly invariant and $C \subseteq \mathcal{E}$. Hence $C \subseteq \mathcal{I}$ by maximality of $\mathcal{I}$.

It remains to show the trajectory enters $\mathcal{I}$ in finite time. By definition, $\mathcal{I}$ is the union of all weakly invariant subsets of $\mathcal{E}$. The cycle set $C$ is one such subset, and the trajectory enters $C$ at time $T_{1}$. Therefore $\Pi_{t} \in \mathcal{I}$ for all $t \geq T_{1}$.

Finally, the tail $\left\{\Pi_{t}\right\}_{t \geq T_{1}}$ is periodic with period $p$; if $p=1$ it is a fixed point, otherwise it is a finite limit cycle.    
\end{pf}

\section{Invariance for the Merge-Split Algorithm}
\label{sec:algorithm_invariance}
We introduce two scalar functionals on the space of partitions $\mathcal{P}(N)$. The first measures violations of Shapley fairness, while the second captures aggregate coalition surplus.

\textbf{Fairness violation.} Define the Shapley fairness violation measure

\[
\Psi(\Pi):=\sum_{S \in \Pi} \sum_{i \in S} \max \left\{0,-\phi_{i}\left(S,\left.v\right|_{S}\right)\right\}, 
\]

which aggregates the total magnitude of negative Shapley values across all coalitions in the partition.

\textbf{Aggregate surplus.} Define the total coalition value

\[
\Phi(\Pi):=\sum_{S \in \Pi} v(S) . 
\]

\begin{definition}[Vector Lyapunov Function]
Define the vector-valued Lyapunov function 
\[
\mathcal{V}(\Pi):=(-\Psi(\Pi), \Phi(\Pi)) \in \mathbb{R}^{2}
\]
equipped with the lexicographic order $\succeq$.    
\end{definition}

\begin{assumption}[Nonnegative singleton values]
\label{assum:1}
Let $v: 2^{N} \rightarrow \mathbb{R}$ be a characteristic function satisfying $v(\{i\}) \geq 0 \text { for all } i \in N .$    
\end{assumption}

\begin{lemma}[Fairness monotonicity under Split Rule]
\label{lemma:Fairness monotonicity}
Under Assumption \ref{assum:1}, the split operator $F^{\mathrm{S}}$ induced by Split Rule satisfies $\sum_{T \in \mathcal{S}(S)} \sum_{j \in T}\left(-\phi_{j}\left(T,\left.v\right|_{T}\right)\right)^{+} \leq \sum_{i \in S}\left(-\phi_{i}\left(S,\left.v\right|_{S}\right)\right)^{+} \text{ for every coalition } S \subseteq N .$ Consequently, for any partition $\Pi \in \mathcal{P}(N)$, $\Psi\left(F^{\mathrm{S}}(\Pi)\right) \leq \Psi(\Pi)$ and the vector Lyapunov function $\mathcal{V}(\Pi)=(-\Psi(\Pi), \Phi(\Pi))$ satisfies $\mathcal{V}\left(F^{\mathrm{S}}(\Pi)\right) \succeq \mathcal{V}(\Pi)$. 
\end{lemma}

\begin{pf}
Fix a coalition $S \subseteq N$ and consider its contribution to the total negative Shapley mass before and after the split induced by Split Rule.

Step 1: Contribution before the split. By definition of $P_{+}(S)$ and $P_{-}(S)$,
\[
\sum_{i \in S}\left(-\phi_{i}\left(S,\left.v\right|_{S}\right)\right)^{+}=\sum_{i \in P_{-}(S)}\left(-\phi_{i}\left(S,\left.v\right|_{S}\right)\right),
\]

since $\phi_{i}\left(S,\left.v\right|_{S}\right) \geq 0$ for all $i \in P_{+}(S)$.

Step 2: Contribution of singleton coalitions. Under Split Rule, every agent $i \in P_{-}(S)$ is isolated as a singleton. By Assumption \ref{assum:1}, $\phi_{i}(\{i\}, v)=v(\{i\}) \geq 0,$ and therefore singleton coalitions contribute no negative Shapley mass: $\sum_{i \in P_{-}(S)}\left(-\phi_{i}(\{i\}, v)\right)^{+}=0$

Step 3: Contribution of the residual coalition. Let $R:=P_{+}(S)$. Split Rule explicitly tests whether the residual coalition $R$ is internally Shapley-fair. 

If $\sum_{j \in R}\left(-\phi_{j}\left(R,\left.v\right|_{R}\right)\right)^{+}>0$, then the rule fully decomposes $S$ into singleton coalitions, in which case Step 2 applies to all agents and the post-split negative Shapley mass is zero. Otherwise, $R$ is retained as a coalition and satisfies $\phi_{j}\left(R,\left.v\right|_{R}\right) \geq 0 \text { for all } j \in R,$ implying $\sum_{j \in R}\left(-\phi_{j}\left(R,\left.v\right|_{R}\right)\right)^{+}=0 .$

Thus, under Split Rule, no residual coalition produced by the split phase contributes negative Shapley mass.

Step 4: Coalition-level comparison. Combining Steps 2 and 3, we obtain $\sum_{T \in \mathcal{S}(S)} \sum_{j \in T}\left(-\phi_{j}\left(T,\left.v\right|_{T}\right)\right)^{+}=0$. Since the pre-split contribution satisfies $\sum_{i \in S}\left(-\phi_{i}\left(S,\left.v\right|_{S}\right)\right)^{+} \geq 0$ the desired inequality follows.

Step 5: Global conclusion. Applying the above inequality to every coalition $S \in \Pi$ and summing over all coalitions yields $\Psi\left(F^{\mathrm{S}}(\Pi)\right) \leq \Psi(\Pi)$. Equivalently, $-\Psi\left(F^{\mathrm{S}}(\Pi)\right) \geq-\Psi(\Pi)$

Since the split phase does not decrease total surplus $\Phi(\Pi)$, the vector Lyapunov function $\mathcal{V}(\Pi)=(-\Psi(\Pi), \Phi(\Pi))$ is lexicographically nondecreasing under the split operator $\mathcal{V}\left(F^{\mathrm{S}}(\Pi)\right) \succeq \mathcal{V}(\Pi)$    
\end{pf}

\begin{lemma}[Monotonicity of the Merge-Split Update]
\label{lemma:Monotonicity}
Consider the merge-split dynamics defined by the update map $F=F^{\mathrm{M}} \circ F^{\mathrm{S}}$, where the split operator $F^{\mathrm{S}}$ is implemented according to Split Rule. Assume that the characteristic function satisfies $v(\{i\}) \geq 0$ for all $i \in N$ (Assumption \ref{assum:1}).

Then, for any partition $\Pi \in \mathcal{P}(N)$ and any $\Pi^{\prime} \in F(\Pi)$ generated by one full merge-split iteration, $\mathcal{V}\left(\Pi^{\prime}\right) \succeq \mathcal{V}(\Pi)$, where $\mathcal{V}(\Pi):=(-\Psi(\Pi), \Phi(\Pi))$ is ordered lexicographically.

Moreover, if $\Pi$ admits a strictly profitable merge, i.e., there exist distinct coalitions $A, B \in \Pi$ such that $v(A \cup B)>v(A)+v(B)$, then $\mathcal{V}\left(\Pi^{\prime}\right) \succ \mathcal{V}(\Pi) \quad \text { for all } \Pi^{\prime} \in F(\Pi) .$    
\end{lemma}

\begin{pf}
Recall that $\mathcal{V}(\Pi)=(-\Psi(\Pi), \Phi(\Pi))$, where $\Psi(\Pi)$ denotes the total negative Shapley mass and $\Phi(\Pi)=\sum_{S \in \Pi} v(S)$ is the total coalition surplus. Lexicographic ordering compares $-\Psi$ first and $\Phi$ only when $-\Psi$ is unchanged. Fix an arbitrary partition $\Pi$ and define $\widetilde{\Pi}:=F^{\mathrm{S}}(\Pi), \quad \Pi^{\prime}:=F^{\mathrm{M}}(\widetilde{\Pi}) .$

Effect of the split phase. By Lemma \ref{lemma:Fairness monotonicity}, which applies under Assumption \ref{assum:1} and Split Rule, the split phase weakly decreases the total negative Shapley mass $\Psi(\widetilde{\Pi}) \leq \Psi(\Pi), \text { equivalently,} -\Psi(\widetilde{\Pi}) \geq-\Psi(\Pi)$. In particular, all singleton coalitions created during the split phase are Shapley-fair, and any residual coalition retained by Split Rule is internally Shapley-fair by construction.

Effect of the merge phase. The merge operator $F^{\mathrm{M}}$ only aggregates coalitions when doing so strictly increases total surplus. Hence, $\Phi\left(\Pi^{\prime}\right) \geq \Phi(\widetilde{\Pi})$, with strict inequality whenever a profitable merge is executed. The merge phase does not alter $-\Psi$ within the same iteration.

Lexicographic comparison. Combining the two phases yields $-\Psi\left(\Pi^{\prime}\right) \geq-\Psi(\widetilde{\Pi}) \geq-\Psi(\Pi) \text{ and } \Phi\left(\Pi^{\prime}\right) \geq \Phi(\widetilde{\Pi}) .$ Therefore, $\mathcal{V}\left(\Pi^{\prime}\right) \succeq \mathcal{V}(\Pi)$.

If $\Pi$ admits a strictly profitable merge, then either: (i) the split phase is inactive and the merge phase strictly increases $\Phi$, or (ii) the split phase weakly increases $-\Psi$ and the merge phase strictly increases $\Phi$. In both cases, $\mathcal{V}\left(\Pi^{\prime}\right) \succ \mathcal{V}(\Pi)$.

This establishes lexicographic monotonicity of $\mathcal{V}$ along the merge-split dynamics, with strict improvement whenever a surplus-positive merge occurs.  
\end{pf}

\begin{theorem}[Invariance for Merge-Split Dynamics]
\label{theorem:main}
Consider the merge-split dynamics induced by the set-valued update map $F=F^{\mathrm{M}} \circ F^{\mathrm{S}}$ on the finite partition space $\mathcal{P}(N)$. Suppose the following assumptions hold:\\
(A1) The characteristic function satisfies $v(\{i\}) \geq 0$ for all $i \in N$ (Assumption \ref{assum:1}).\\
(A2) The split operator $F^{\mathrm{S}}$ is implemented according to Split Rule, isolating agents with negative Shapley values and fully decomposing coalitions whenever residual fairness violations persist.\\
(A3) The Lyapunov function $\mathcal{V}(\Pi)=(-\Psi(\Pi), \Phi(\Pi))$ is lexicographically nondecreasing along admissible merge-split updates, and strictly increases whenever a strictly profitable merge is executed (Lemma \ref{lemma:Monotonicity}).

Then every trajectory $\left\{\Pi_{t}\right\}_{t \geq 0}$ with $\Pi_{t+1} \in F\left(\Pi_{t}\right)$, starting from an arbitrary initial partition $\Pi_{0} \in \mathcal{P}(N)$, enters in finite time the largest weakly invariant set

\[
\mathcal{I}:=\left\{\Pi \in \mathcal{P}(N) \mid \exists \Pi^{\prime} \in F(\Pi) \text { such that } \mathcal{V}\left(\Pi^{\prime}\right)=\mathcal{V}(\Pi)\right\} .
\]

Moreover, the invariant set $\mathcal{I}$ consists of:\\
(a) fixed points corresponding to Shapley-fair and merge-stable partitions; and\\
(b) finite limit cycles generated by value-preserving transitions, induced by persistent alternation between fairness-driven splits and surplus-driven merges.
    
\end{theorem}

\begin{pf}
 Consider any trajectory $\left\{\Pi_{t}\right\}_{t \geq 0}$ generated by the set-valued dynamics $\Pi_{t+1} \in F\left(\Pi_{t}\right)$.

Step 1: Lyapunov monotonicity along trajectories. By Assumption (A3), the Lyapunov sequence $\left\{\mathcal{V}\left(\Pi_{t}\right)\right\}_{t \geq 0}$ is lexicographically nondecreasing along every admissible transition. Moreover, $\mathcal{V}\left(\Pi_{t+1}\right) \succ \mathcal{V}\left(\Pi_{t}\right)$ whenever $\Pi_{t}$ admits a strictly profitable merge.

Step 2: Finite-time stabilization of Lyapunov values. The partition space $\mathcal{P}(N)$ is finite. Consequently, the image set $\mathcal{V}(\mathcal{P}(N))$ is finite. A lexicographically nondecreasing sequence taking values in a finite set must stabilize in finite time. Hence, there exists a finite time $T$ such that $\mathcal{V}\left(\Pi_{t+1}\right)=\mathcal{V}\left(\Pi_{t}\right) \text { for all } t \geq T.$

Step 3: Entry into the weakly invariant set. By definition of $\mathcal{I}$, the equality $\mathcal{V}\left(\Pi_{t+1}\right)=\mathcal{V}\left(\Pi_{t}\right)$ implies that $\Pi_{t} \in \mathcal{I}$. Therefore, every trajectory enters $\mathcal{I}$ in finite time.

Step 4: Structure of the invariant set. Let $\Pi \in \mathcal{I}$. If $\Pi$ is a fixed point of $F$, then no admissible split and no strictly profitable merge exists. Under Split Rule and Assumption \ref{assum:1}, this is equivalent to $\Pi$ being both Shapley-fair and merge-stable, establishing part (a).

If $\Pi$ is not a fixed point, then there exists at least one successor $\Pi^{\prime} \in F(\Pi)$ such that $\mathcal{V}\left(\Pi^{\prime}\right)=\mathcal{V}(\Pi)$. Since strict Lyapunov improvement is ruled out, such transitions must preserve both the total negative Shapley mass $\Psi$ and the total surplus $\Phi$. Because $\mathcal{P}(N)$ is finite, repeated application of value-preserving transitions necessarily produces a finite cycle, establishing part (b).    
\end{pf}

\section{Case Study: Merge-Split Dynamics}
\label{sec:case_study}
We consider a transferable-utility cooperative game ( $N, v$ ) with $N=\{1, \ldots, 10\}$. Let $C= \{1,2,3,4\}$ denote a subset of productive players. The characteristic function $v: 2^{N} \rightarrow \mathbb{R}$ is constructed to exhibit heterogeneous incentives and fairness tensions: players in $C$ have higher standalone value than those outside $C$, coalitions contained entirely in $C$ generate increasing surplus, and coalitions containing all members of $C$ receive an additional synergy bonus. As a result, the game violates marginal monotonicity and induces negative Shapley allocations for certain players in mixed coalitions. To align with Assumption \ref{assum:1} and the invariance results of Theorem \ref{theorem:main}, the characteristic function is additively shifted so that $v(\{i\}) \geq 0$ for all $i \in N$. This shift preserves Shapley signs, merge incentives, and all ordinal comparisons relevant to the merge-split dynamics.

At iteration $t$, the system is in a coalition structure $\Pi_{t}=\left\{S_{1}, \ldots, S_{k}\right\}$. For each coalition $S \in \Pi_{t}$, let $\phi_{i}(S)$ denote the Shapley value of player $i$ in the restricted game $\left.v\right|_{S}$. The dynamics follow the merge-split algorithm introduced in Section \ref{sec:dynamics} with the split phase implemented according to Split Rule:

\begin{itemize}
  \item (Split) Any agent with $\phi_{i}(S)<0$ is isolated as a singleton. The residual coalition is retained only if it is internally Shapley-fair; otherwise, the coalition is fully decomposed into singletons.
  \item (Merge) If no split is applied, any pair of coalitions $A, B \in \Pi_{t}$ satisfying $v(A \cup B)>v(A)+ v(B)$ may merge.
  \item Exactly one split or one merge is applied per iteration.
\end{itemize}

The dynamics are analyzed using the vector Lyapunov function $\mathcal{V}(\Pi)=(-\Psi(\Pi), \Phi(\Pi))$. Here, $\Phi(\Pi)= \sum_{S \in \Pi} v(S)$ denotes the total coalition surplus, while fairness violation term is defined as $\Psi(\Pi)=\sum_{S \in \Pi} \sum_{i \in S} \max \{0,-\phi_{i}(S)\}$.

Starting from a randomly generated initial partition, the dynamics rapidly eliminate fairness violations through singleton splits. In the run reported here, a single split operation at $t=0$ isolates all agents with negative Shapley values and yields the partition $\Pi_{1}=\{\{2,3,4\}\} \cup\{\{i\} \mid i \in N \backslash\{2,3,4\}\}$. This partition is internally Shapley-fair and admits no strictly profitable merge. Consequently, the system enters the invariant set $\mathcal{I}$ at $t=1$ and remains there for all subsequent iterations. Although the algorithm continues to be iterated, neither the coalition structure nor the Lyapunov value changes thereafter, illustrating finite-time entry into the invariant set rather than asymptotic convergence.\\

\begin{figure}[h]
\begin{center}
  \includegraphics[width=0.5\textwidth]{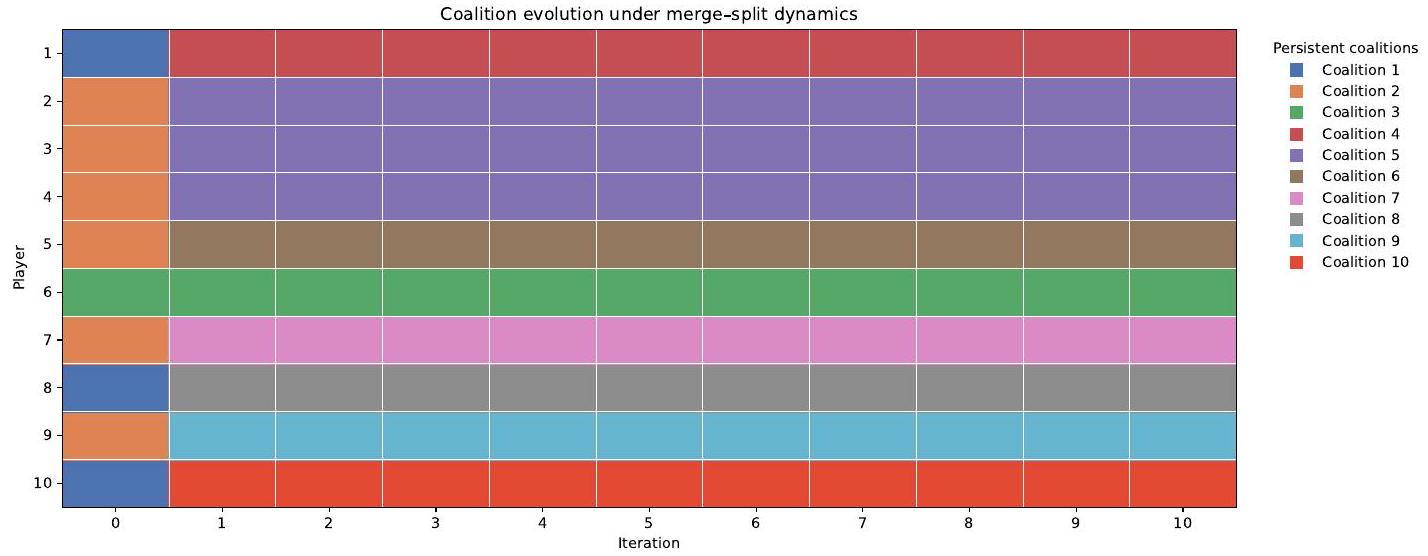}
\caption{Coalition evolution under the merge-split dynamics. After a single split iteration, the system enters the invariant set and the coalition structure remains unchanged.}
\end{center}
\end{figure}

Figure 1 visualizes coalition membership over time. Rows correspond to players and columns to iterations; colors denote coalitions. After the initial transient, coalition membership becomes constant, reflecting stabilization within the invariant set.

\begin{figure}[h]
\begin{center}
  \includegraphics[width=0.45\textwidth]{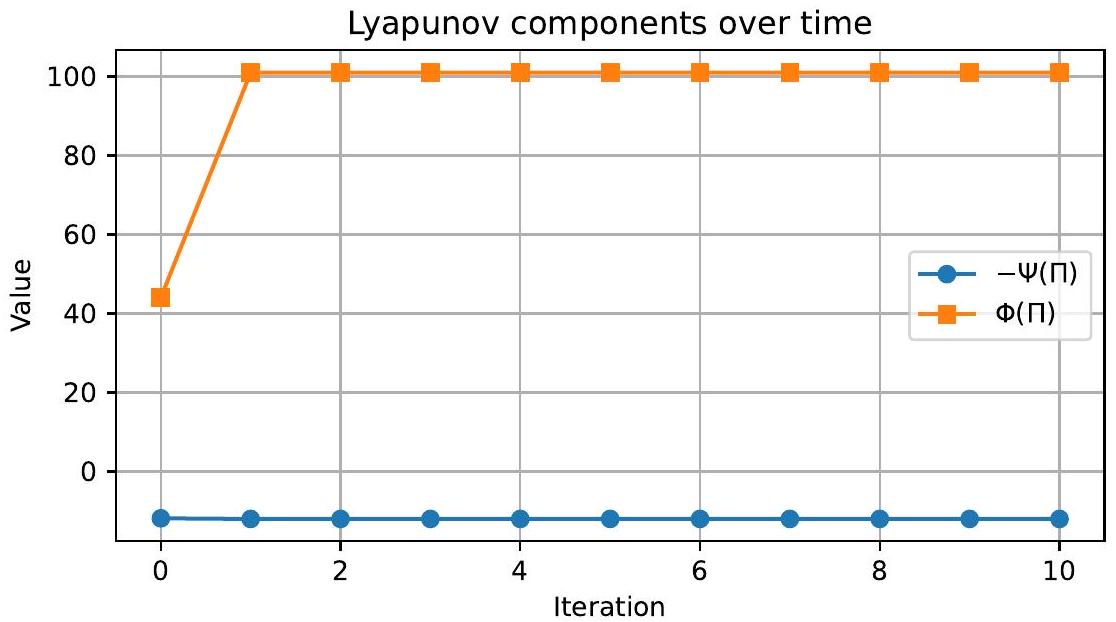}
\caption{Evolution of the Lyapunov components. The vector Lyapunov function stabilizes after finite time, confirming finite-time entry into the invariant set.}
\end{center}
\end{figure}

Figure 2 shows the evolution of the Lyapunov components $-\Psi\left(\Pi_{t}\right)$ and $\Phi\left(\Pi_{t}\right)$. The Lyapunov sequence is lexicographically nondecreasing and stabilizes in finite time.

This case study illustrates the role of Split Rule in enforcing fairness monotonicity and preventing oscillatory behavior. Fairness violations are resolved through targeted singleton isolation, and once all surplus-enhancing merges are exhausted, the dynamics remain within the invariant set. In this example, the invariant set consists of a single Shapley-fair and mergestable fixed point. More generally, as shown in Theorem \ref{theorem:main}, the invariant set may also contain finite cycles when value-preserving transitions exist.

\section{Conclusion}
\label{sec:conclusion}
This paper presents a dynamic control-theoretic framework for coalition formation that harmonizes the often competing objectives of fairness and efficiency. By assigning distinct roles to topological operators—using Shapley value fairness to drive splits and collective surplus to drive merges—we establish a rigorous pathway to equilibrium. The proposed Fairness-Stable split rule is proven to eliminate ``agent-responsible inefficiency,'' ensuring that corrective actions are targeted and non-regressive. Our theoretical analysis, grounded in a vector Lyapunov approach and the LaSalle invariance principle, guarantees finite-time convergence to a set of Shapley-Fair and Merge-Stable (SFMS) partitions. These results bridge the gap between static cooperative game theory and dynamic system stability. Future work will extend this framework to stochastic environments where the characteristic function is learned online from noisy observations, and explore applications in large-scale systems such as LLM agent collaboration networks.

\bibliography{ifacconf}

\end{document}